\documentclass[reqno, 11pt, final]{amsart}\usepackage{graphicx}
 \newtheorem{theorem}{Theorem}[section]
   
 \newtheorem{prop}[theorem]{Proposition} 
  \newtheorem*{acknowledgement*}{Acknowledgements}%

\def\res{\mathop{\hbox{\vrule height 7pt width.5pt depth 0pt\vrule height .5pt width 6pt depth 0pt}}\nolimits}

\def\res{\mathop{\hbox{\vrule height 7pt width.5pt depth 0pt\vrule height .5pt width 6pt depth0pt}}\nolimits}

\def\argmin{\hbox{argmin}}

 \theoremstyle{definition}
 
 \theoremstyle{remark}

 \newcommand{\be}{\begin{equation}}
\newcommand{\disp}{\displaystyle}
\newcommand{\ee}{\end{equation}}
\newcommand{\ba}{\begin{array}}
\newcommand{\ea}{\end{array}}
\newcommand{\vk}{\vspace {.3cm}}

\newcommand{\vks}{\vspace {.1cm}}
\newcommand{\vkk}{\vspace {.6cm}}

\newcommand{\hk}{\hspace {.3cm}}

\newcommand{\hkk}{\hspace {.6cm}}
\newcommand{\hkkk}{\hspace {.9cm}}
\def\real{I\!\!R}

\newcommand{\uz}  {\mbox{\boldmath{$w$}}}
\newcommand{\uuz}  {\mbox{\boldmath{$u$}}}
\newcommand{\iz}  {\mbox{\boldmath{$i_1$}}}
\newcommand{\vz}  {\mbox{\boldmath{$v$}}}

\newcommand{\unoz}  {\mbox{\boldmath{$1$}}}

\newcommand{\Dz}  {\mbox{\boldmath{$D$}}}
\newcommand{\Iz}  {\mbox{\boldmath{$I$}}}
\newcommand{\Az}  {\mbox{\boldmath{$A$}}}
\newcommand{\Bz}  {\mbox{\boldmath{$B$}}}
\newcommand{\Ba}  {\mbox{\boldmath{$B$}}_1}
\newcommand{\Bb}  {\mbox{\boldmath{$B$}}_2}
\newcommand{\Jd}  {J_{\mbox{\it \tiny D}}}

\def\argmin{\mathop{{\rm argmin}}\nolimits}
\def\disp{\displaystyle}

\def\res{\hbox{\vrule height8pt\vrule height0.4pt width8pt\kern1pt}}

\def\sbv {S\kern-1.5pt BV}

\def\sbh {S\kern-2pt B\kern-2pt H}

  \def\s{\sigma}   
 
\def\n{{\s}} 

\title{Mechanics of reversible unzipping}

 \author[]{ F. MADDALENA, D. PERCIVALE, G. PUGLISI,
 L. TRUSKINOVSKY}
  \address{ Dipartimento di Matematica Politecnico di Bari,
  via Re David 200, 70125 Bari, Italy}
  \address{Dipartimento di Ingegneria della Produzione
  Termoenergetica e Modelli Matematici, Universit\`{a}
  di Genova, Piazzale Kennedy, Fiera del Mare, Padiglione D, 16129 Genova, Italy}
\address{
Dipartimento di Ingegneria Civile e Ambientale Politecnico di Bari,
  via Re David 200, 70125 Bari, Italy}
  \address{Laboratoire de Mechanique des Solides,
CNRS-UMR 7649,
  Ecole Polytechnique,
91128, Palaiseau, France \vkk}
 \email{f.maddalena@poliba.it, percivale@diptem.unige.it,
  g.puglisi@poliba.it,\newline trusk@lms.polytechnique.fr}

 \subjclass{}
\begin{document}

\begin{abstract}
We study the mechanics of a reversible decohesion (unzipping) of an elastic layer subjected to quasi-static
end-point loading. At the micro level the system is simulated by an elastic chain
of particles interacting with a rigid foundation through breakable springs.
Such system can be viewed as prototypical for the description of a wide range of
phenomena from peeling of polymeric tapes, to rolling of cells, working of
gecko's fibrillar structures and denaturation of DNA. We construct a rigorous continuum limit of the discrete model which captures both stable and metastable configurations and present a detailed parametric study of the interplay between elastic and cohesive interactions. We show that the model reproduces the experimentally observed abrupt transition from an
incremental evolution of the adhesion front to a sudden complete decohesion of a
macroscopic segment of the adhesion layer. As the microscopic parameters vary
the macroscopic response changes from quasi-ductile to quasi-brittle, with
corresponding decrease in the size of the adhesion hysteresis. At the
micro-scale this corresponds to a transition from a `localized' to a `diffuse'
structure of the decohesion front (domain wall). We obtain an explicit expression for the critical debonding threshold in the limit when the internal length scales are much smaller than the size of the system.  The achieved parametric control of the
microscopic mechanism can be used in the design of new biological inspired
 adhesion devices and machines.
 \vkk

\noindent Key words: Unzipping, Adhesion, Peeling, Hysteresis, DNA, Gecko, Calculus of Variations,
$\Gamma$-convergence.

\end{abstract}

\maketitle

\section*{Introduction}
\setcounter{equation}{0}

Adhesion phenomena are governed by complex energy exchanges between multiple
scales representing hierarchical structures. The phenomenological modeling of
adhesion has been successful in describing a
variety of experimentally observed static and dynamic regimes (see
\cite{Frem, Ken, PPG, CG}). The phenomenological models, however, are of a
black box type and have limited predictive power outside of a particular range
of parameters. More importantly, they can not be used for the microstructural
optimization of the artificially created adhering materials and nanorobotics devices. It is therefore of
interest to develop a multi-scale approach linking the microscopic
attachment-detachment mechanisms with the macroscopic phenomenological
parameters. This step is crucial for the analysis of a variety of adhesion
related phenomena from fiber decohesion in composites \cite{MP, MT} and
crazing/peeling phenomena in polymers \cite{kramer,Ken}, to the activity of focal
adhesions involved in cell motility \cite{Des,MHG,DTSH} and  the low temperature denaturation of single molecule DNA \cite{Pey}.

In this paper we contribute to this general task by considering the minimal
 model accounting for the interplay between elasticity and cohesion.  While this model captures only the most important effects
associated with quasi-static decohesion, it has the advantage of being amenable
to a completely transparent mathematical analysis allowing one to study how
macroscopic responses vary depending on the microscopic parameters.  We
focus on the case when the debonding is reversible. This type of decohesion, also known as \emph{unzipping}, is crucial for the working
of  a variety of biological systems (\cite{SG}), in particular, for the
functioning of the self-similar fibrillar Gecko's hair(\cite{JB,P,YG}). The irreversible version of the model, which is more adequate for the description of decohesion in conventional engineering
materials, will be presented in a separate paper. For a general comparison of reversible and irreversible fracture see \cite{DT}.

\begin {figure}[h]\vspace{0 cm}
\includegraphics[height=4.5 cm]{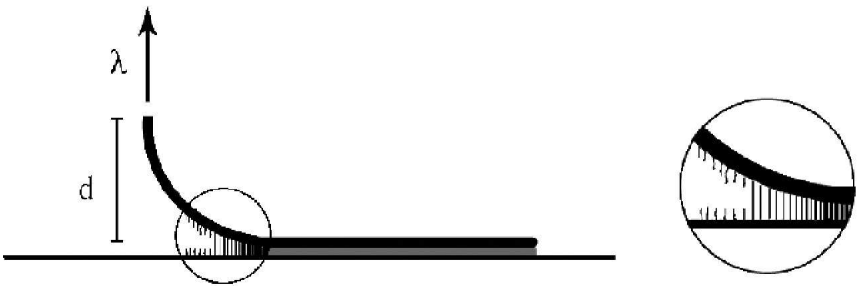}
\caption{\setlength{\baselineskip}{13 pt}{\footnotesize Finite elastic layer attached to a rigid background. The controlling parameter is the displacement $d$ applied at one end, the other end is free.  The measured response is represented by the dependence of the point force $\lambda$ on $d$.}}
\label{schemeo}\end{figure}

According to typical observations the process of
reversible  decohesion due to quasi-static point loading includes three main stages (\cite{Ken} and references therein). First the system
behaves elastically until decohesion begins. Then there is a steady state incremental
evolution of the decohesion front. Finally, at a critical threshold, the system
undergoes a sudden transition to the fully debonded configuration. Similarly, if the
system is unloaded from the fully debonded state, there exists an unloading threshold
beyond which a finite portion of the adhesive layer suddenly reattaches to the
adhesive background. The whole phenomenon is usually hysteretic with different
attachment and detachment thresholds.

In an attempt to reproduce this basic behavior we consider a chain of massless points connected
by harmonic shear springs. The particles are  attached to a
rigid support by breakable elastic links (see Fig.\ref{schemeo}) which mimic,
depending on the  parameters, either direct molecular interactions or
interactions through the fibrillar adhesive layer with internal elasticity. For simplicity we neglect the bending stiffness of the elastic layer
which could be accounted for by adding elastic interaction of next to nearest neighbors \cite{MP,OF}.
We consider a quasi-static transversal point loading of the otherwise free layer in a
hard device and study the rate independent evolution of the emerging debonding
front which can be viewed as a \emph{domain wall} separating bonded and debonded phases \cite{Pey}. The reversal of the front direction represent the switch between
zipping and unzipping.

Similar discrete models of the Frenkel-Kontorova
type have been used previously in the analysis of lattice trapping of cracks in crystals
\cite{KT}, interfacial wetting \cite{Lip} and other phenomena where an on site potential with sublinear growth competes with an elastic coupling of the elements \cite{BK}.   In connection to the duplication and transcription
of the DNA an approach of this type was first proposed by Peyrard and Bishop who applied it to the modeling of equilibrium melting phase transition (see the review \cite{Pey}). Our model can be viewed as a purely mechanical version of the Peyrard-Bishop model where we go beyond global minima of the energy (zero temperature limit) and investigate the structure of the whole energy landscape (see also \cite{Theo}).
Our use of the simplified piece-wise
quadratic approximation for the on site potential allows us to formulate the main results in the analytic form.

We use an incremental energy minimization approach
and solve the finite dimensional variational problem for each value of applied
displacement. Due to the simplicity of the cohesive potential we are able to find all equilibrium
configurations and identify those representing global and local minima of
the energy. Our analysis shows that the local minimizers of the energy always have a single decohesion front.
The metastable configurations, forming separate branches parameterized by
the loading parameter, can abruptly end. The absence of continuity leads to
the necessity of `dynamic snapping' from one branch
to another. While these events may be
 dissipative, they  do not prevent overall reversibility (see also \cite{JB,P, YG}).

We discuss two evolutionary strategies. One strategy assumes  an
overdamped gradient flow dynamics and can be viewed as a vanishing viscosity
limit (maximum delay convention, e.g. \cite{PT7,DT}). The other strategy assumes that the system always remains in the global minimum of the energy (Maxwell convention). For these two models we establish the
existence of the thresholds separating the regime of incremental propagation of the
decohesion front from the regime of a sudden and massive decohesion representing a size effect. When we
restrict the evolution of the system to the global minimization of the energy, the
loading and unloading thresholds coincide. When instead we allow the system to
follow the maximum delay convention and explore the set of marginally stable
configurations, the two thresholds become different. The comparison with
experiments shows that it is the vanishing viscosity solution which reproduces
the observed adhesion hysteresis most faithfully (\cite[Chapter 3]{Ken}).

We then develop a continuum analog of our finite dimensional microscopic model, interpreting it
as a formal $\Gamma$-limit \cite{Bra,DM}. While the limiting functional,
constructed in this way, usually captures only the global minima in the
original discrete problem, in our case it also agrees with a point wise limit and therefore preserves the local
minima. To prove this fact we perform a systematic study of all metastable solutions
of the continuum problem. While the continuum model is much more transparent
mathematically and allows one to obtain the values of
all relevant thresholds in explicit form, the strongly discrete limit remains important for
some applications, in particular, for the modeling of the DNA \cite{Pey}.

The primary goal of our simplified model is to elucidate how macroscopic
responses  depend on the microscopic parameters. The continuum version of the model
contains only one non-dimensional micro-parameter $\nu$ measuring the propensity of the adhesive layer to dynamic snapping. We show
that the degree of localization of the decohesion front increases as $\nu$
decreases which is revealed macroscopically as a transition from quasi-ductile
to quasi-brittle behavior. In the limit $\nu\rightarrow 0 $, which corresponds to the disappearing of an internal length scale, we obtain an explicit expression for the critical debonding force, which is in principle a measurable parameter \cite{Theo}.  In general, we expect that the achieved parametric control in the simplified microscopic setting can be used in the design of the prototypical molecular
devices. We are fully aware, however, that there is long way between the toy models of the type considered in this paper and the realistic description of the particular biological systems (gecko, DNA, etc.)

The paper is organized as follows. In Section~\ref{ca} we introduce our
discrete model and present an analytical description of all stable and
metastable configurations corresponding to a given value of applied displacement. In
Section~\ref{cont} we derive the $\Gamma$-limit of the discrete model and
classify the local minimizers of the limiting problem. In Section~\ref{mdcr1}
we study two different responses of the discrete model to monotonous and cyclic
loading, one overdamped and another nondissipative. In the dissipative case we
compute the associated heat to work ratio and construct the hysteresis loops.
Finally, in Section~\ref{fin} we show how the main features of the
cohesion/decohesion hysteresis vary as one goes from discrete to continuum
description and present the results of the detailed parametric study of the
model. In the Appendices we collect mathematical results of technical nature.

\section{Microscopic model} \label{ca}
\setcounter{equation}{0}

Consider a discrete chain containing $n+1$ points which are connected by linear
elastic springs with reference length $l=L/n$. Each point is also
connected to a rigid substrate by a breakable spring. In this maximally
simplified setting one can deal with two prototypical problems: pull out test
(e.g. \cite{MT}) and pull off or peeling test (e.g. \cite{OF}). For
determinacy, we shall focus on the peeling problem and assume that the points
move orthogonally to the substrate (see Fig.\ref{scheme}).

\begin {figure}[h]\vspace{0.5 cm}
\includegraphics[height=4.5 cm]{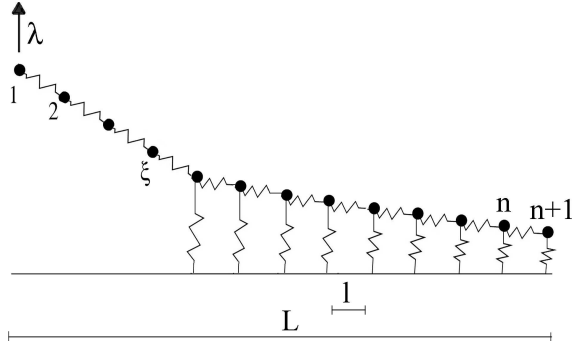}
\caption{\setlength{\baselineskip}{13 pt}{\footnotesize Schematic representation of
the discrete model. The location of the decohesion front is marked by the integer valued parameter $\xi$. In the horizontal direction the particles are separated by the fixed distance $l=L/n$. }} \label{scheme}\end{figure}

Denote by  $u_i$ the vertical displacements of the particles from their reference
positions. The elastic energy of the connecting linear springs can be written as
 \be\displaystyle \phi(\delta_i)=\frac{1}{2} G \, \delta_i^2,
\label{ffa} \ee
 \noindent where $G$ is the ({\it shear}) modulus and
$\delta_i=(u_{i+1}-u_{i})/l.$ For the energy of the breakable springs we assume the
simplest form
 \be \varphi(w_i )=\left \{ \ba{ll} \displaystyle
\frac{1}{2} \, E w_i^2 &\mbox{if}\:\: \vk w_i< 1 \\
\displaystyle \frac{1}{2}  E  &\mbox{if}\:\: w_i\geq 1, \ea \right  . \label{ffc} \ee
\noindent where $E$ is the longitudinal elastic modulus, $w_i=u_i/u_r$ are the
normalized displacements, and $u_r$ is the
 breaking threshold. The total energy of the chain can be written as
 \be \Phi=\frac{L}{n} (\sum_{i=1}^{n+1}  \varphi(w_i)+ \sum_{i=1}^{n}
  \phi (\delta_i)). \label{ener}\ee
We load the chain in a hard device, meaning that the (normalized) displacement $d$ of
the first spring is prescribed
 \be w_1=d>0.\label{bc}\ee

The energy can be rewritten in a more compact form.  To this end we introduce a
vector distinguishing `bonded' and `de-bonded' springs
 \be \chi(i)=\left \{ \ba{l} \displaystyle
0 \hk \mbox{if} \hk w_i <1\vk \\
1 \hk \mbox{if } \hk w_i\geq 1, \ea \right . \label{a}\ee and construct the
$(n+1) \times (n+1)$ diagonal matrix $$\Dz=\mbox{diag}(\chi(1), \ldots,
\chi(n+1)).$$
 In what follows it will be convenient to use the following notations: $\uz\in \real^{n+1}$ - the
displacement vector, $\unoz\in \real^{n+1}$ - the vector with $(\unoz)_i=1$, and $\iz$
- the first vector of the canonical basis in $\real^{n+1}$. We also introduce the
$(n+1) \times (n+1)$ tri-diagonal matrix
 $$
          \Az=\left [ \ba{lllllll}
              1& -1 &  & & & &\mbox{ \large {\bf 0}}\\
              -1& 2 & -1& & & & \\
               & -1 & 2&-1& & & \\
               &   & \ddots&\ddots&\ddots& & \\
               & &   & -1&2&-1& \\
                &  &   &  &-1&2&-1\\
              \mbox{ \large {\bf 0}}& &   &  & & -1 & 1\\
              \ea \right ].
$$
By using the above notations, we can rewrite the dimensionless total energy (\ref{ener}) in
the form
 \be
 \Jd(\uz):= \frac{\Phi}{L E} =\frac{1}{2 n}(
 \Bz \uz \cdot \uz +\xi),\label{cf}\ee
 where
 \be \Bz=\Iz -\Dz
  +n^2\nu^2\, \Az
 \label{B}\ee
 and $\xi$ is the total number of debonded springs. The dimensionless
energy (\ref{cf}) depends on two scaling parameters: $n$ and
 \be \nu=\frac{u_r }{L }\sqrt{\frac{G}{E}}. \label{tau}\ee
The main physical nondimensional parameter of the problem, $\nu$, implicitly
characterizes the toughness of the breakable bonds. In particular, by
decreasing $\nu$ we increase the cohesion energy. The geometrical parameter $n$
is a measure of discreteness and $n \rightarrow \infty$ would mean for us the
`macroscopic' or `continuum'  limit (see Section~\ref{cont}).

To find the equilibrium state of the chain at a given $d$, it is natural to
first minimize the elastic energy at a fixed configuration of debonded springs
$\Dz$. We obtain the following minimization problem
 \be \mbox{Min}\left \{\Jd(\uz)\;\vert \;
 \uz \in {\real}^{n+1},
 \uz \cdot \iz=d
 \right \}. \label{min}\ee
\noindent The necessary conditions of equilibrium can be written as
 \be  \Bz \uz -\sigma n \iz={\mathbf 0},\label{eul2}\ee
\noindent where $$\sigma=\frac{u_r}{EL}\lambda$$ is the Lagrange multiplier, representing
the external force $\lambda$ acting on the first point of the chain. The stability of these
equilibrium configurations is ensured by the positive definiteness of the Hessian
matrix $\Bz$ which immediately follows from (\ref{B}). We can then conclude that all
solutions of (\ref{eul2}) are local minima of the energy.

The linear equations (\ref{eul2}) can be solved formally which allows us to represent
all metastable equilibrium configurations by the formulas

 \be  \sigma=  \frac{d}
 {n\Bz^{-1}\iz \cdot \iz}, \label{es1}\ee

 \be  \uz=
 \frac{\Bz^{-1}\iz}{\Bz^{-1}\iz \cdot \iz}d,
 \label{es2}\ee

 \be \Jd=
\frac{1}{2n}\left (\frac{d^2}{\Bz^{-1}\iz \cdot \iz}
 +\xi\right ).\label{es3}\ee
 \smallskip

We observe that the variables $w_i$ given by {\rm(\ref{es2})} decrease as the
index $i$ increases. This follows from the fact that the column elements in the
inversions of diagonally dominant tri-diagonal matrices necessarily decrease
(see e.g. \cite{Nabb}). Therefore in each metastable configuration it is
necessarily the first $\xi$ springs that are debonded while the remaining
$n+1-\xi$  are bonded. This observation allows one to write the following
analytical representation for the displacement field (see \textbf{Appendix A}
for details)
 \be
 \displaystyle w_i=\left \{
 \ba{ll} \displaystyle d - (i-1)\frac{\sigma}{n\nu^2}, &  i=1,...,\xi,\vk \\
\displaystyle \frac{\cosh [(n+3/2-i)\eta]}{\sinh [(n+1-\xi) \eta] \sinh
[\eta/2]}\frac{\sigma}{2n\nu^2} , & i=\xi+1,...,n+1.\ea \right . \label{disp}\ee Here
 \be
 \displaystyle \sigma= K d \label{lambda}\ee
 is the stress,
\be K= \frac{2n\nu^2}{2\xi-1+\coth \frac{\eta}{2}\coth[(n+1-\xi)\eta]}\ee is
the effective elastic modulus,  and $\eta$ is one of the two solutions of the
equation
 \be
\displaystyle 1+\frac{1}{2n^2 \nu^2}= \cosh [\eta].\label{eta}\ee Since
 the equilibrium properties are represented by even functions of
$\eta$,  the particular choice of the solution in (\ref{eta}) is irrelevant.

The energy of the metastable configurations can be written as
 \be \displaystyle J=\bar J(d,\xi)
 = \frac{1}{2n}(n K
 d^2+\xi).\label{eqene}\ee
In order to be admissible, the configurations (\ref{disp}) must respect the
compatibility condition requiring that all bonded springs have $w_i<1$ and all
debonded springs have $w_i\geq1$. This condition is equivalent to a restriction on
$\xi$. To obtain this restriction we compute the value of the loading parameter
$d_1(\xi)$ corresponding to $w_{\xi+1}=1$ and the value $d_2(\xi)$ corresponding to
$w_{\xi}=1$. We obtain
 \be\ba{ll}
 d_1(\xi)=\displaystyle 1+\frac{2(\xi-1)}
 {\coth[\frac{\eta}{2}]\coth[(n-\xi+1)\eta]+1} ,\vks \\
\displaystyle d_2(\xi)=1+\frac{2\xi}
 {\coth[\frac{\eta}{2}]\coth[(n-\xi+1)\eta]-1} . \hkk
\ea
 \label{br}\ee
We call the interval $[d_1(\xi),d_2(\xi)]$ the {\it stability domain} of a solution
with a given geometry of a crack $\xi$. In general, several crack geometries may be
compatible with a given load $d$. The detailed study of the obtained solutions, in
particular the specialization of the global minimizers, will be postponed till
Section \ref{mdcr1}.

\section{Macroscopic problem}\setcounter{equation}{0}
\label{cont}

In most applications the parameter $n$ is a large number. Therefore it is of interest
to describe the continuum limit of the discrete model formulated in the previous
section. As a first step we can look at the point-wise limits of the discrete
solutions (\ref{disp}) as $n\rightarrow\infty$. To compute these limits we introduce
$$X(i)=(i-1)L/n,$$ the coordinate of the $i$th spring and define the following normalized
variables:

$$\ba{ll} x:=X/L, & \mbox{ normalized spatial coordinate,} \vks \\
\zeta:=\xi/n, & \mbox{fraction of debonded springs.} \ea\vks$$
 \noindent By assuming
that in the limit $n\rightarrow\infty$  the parameter $\zeta$ is finite we obtain
from (\ref{lambda})
 \be d(\zeta)=1+\frac{\zeta}{\nu^2}
 \sigma.\label{limd}\ee
It is also easy to see that
 \be\lim_{n\rightarrow\infty}d{_1}(\zeta)=\lim_{n\rightarrow\infty}d{_2}(\zeta)=
  d(\zeta).\label{dddd}
\ee This means that for each value of $\zeta$ the stability domain shrinks in the
continuum limit to a point. For the limiting  displacement field we obtain \be
w_{\zeta}( x)=\left \{ \ba{ll}\displaystyle
 \displaystyle d-\frac{\sigma}{\nu^2} x &
\mbox{ if } x \in (0,\zeta ), \vks \\
\displaystyle\frac{ \cosh(\frac{1}{\nu} (1-x))}{\cosh (\frac{1}{\nu}(1-\zeta) )} &
\mbox{ if } x \in (\zeta , 1) .\ea \right .\label{ucont}\vks\ee The value of the
continuum energy of a metastable state can now be computed from the formula
 \be
J=\hat J(\zeta)=\frac{1}{2} ( \zeta \,
 (1+\frac{\sigma^2}{\nu^2})+ \sigma).\label{J}\ee
Here we used the limiting relation between the stress and the length of the
debonded region \be \sigma=\nu\tanh \left (\frac{1-\zeta}{\nu}\right
).\label{J1}\ee

To interpret these results correctly, we can independently look at the limit of
the variational problem (\ref{min}) as $n\rightarrow\infty$. To this end we can
define the space of piecewise constant functions on $(0,1)$
$$
\disp\mathcal A_n(0,1)=:\left\{\sum_{i=1}^{n} a_i^n{\bf1}_{[i-1,i)\frac{1}{n}}:
a_i^n\in\mathbb R,\ a_1^n=d\right\}
$$
and rewrite the discrete energy functional (\ref{ener}) in the form
$$
\displaystyle
 J_n(w)=\left\{\begin{array}{ll}&\displaystyle
\frac{1}{n E}\left( \displaystyle\sum_{i=1}^{n+1}\varphi\left(w\left
(\frac{i-1}{n} \right ) \right )+\sum_{i=1}^{n}\phi\left( \frac{
w(\frac{i}{n})-w( \frac{i-1}{n})}{1/n}\frac{u_r}{L}\right)
\right)\quad\hbox{if}\quad w\in \mathcal A_n(0,1),\\&\\ & +\infty\quad
\hbox{otherwise in}\quad L^2(0,1).\\\end{array}\right.$$
Next we can define
${\mathcal A}^*_n$ as the subset of the functions $w\in H^1(0,1)$ such that
there exists $\hat w\in\mathcal A_n$ which satisfies
\begin{equation}
\label{link} w'(x)=\disp\sum_{i=1}^{n}\frac{\hat w\left(\frac{i}{n}\right)-\hat
w\left(\frac{i-1}{n} \right)}{1/n}{\bf1}_{[i-1,i)\frac{1}{n}},\ \ w(0)= \,d.
\end{equation}
 Clearly, $
w({\frac{i}{n}})=\,\hat{w}(\frac{i }{n}) $ and we rewrite the original
functional in the form \be
 J_n(w)=\frac{1}{nE}\displaystyle\sum_{i=1}^{n+1}\varphi\left(w\left
(\frac{i}{n} \right )\right)+\frac{1}{E}\int_0^1 \phi(\frac{u_r}{L}w')\,dx
,\label{J11}\ee
 where now $w\in {\mathcal A}^*_n $. It
is easy to see that all (local and global) minimizers of $J_n$ on $\mathcal A_n$ can
be described in terms of the corresponding minimizers of $J_n$ on ${\mathcal A}^*_n$
which makes the two problems equivalent.

We can now study a point-wise limit of the functional (\ref{J11}). It is
straightforward to see that this finite dimensional variational problem converges as
$n\rightarrow\infty$ to the infinite dimensional problem for the continuum functional
\begin{equation}
\label{link1} J (w)=\left \{ \begin{array}{l}\displaystyle \frac{1}{E }
\int_0^1
(\varphi (w)+\phi(\frac{u_r}{L} w'))\, dx \mbox{ if } w\in \mathcal A(0,1)\vspace{11 pt}\\
\displaystyle +\infty \mbox{ otherwise in } L^2(0,1)\end{array} \right .
\end{equation}
\noindent which is defined on the space $ \mathcal A=\{ w\in H^1(0,1): w(0)=d\}.$
 In $\textbf{Appendix B}$ we prove that the
point-wise convergence automatically implies $\Gamma$- convergence. In
particular, this guarantees that the global minimizers of (\ref{link1}) can be
viewed as the continuum limits of the global minimizers of (\ref{J11}).

The next question concerns the status of the local minimisers of (\ref{link1}).
We say that $w\in \mathcal A$ is a local minimizer of $J$ if there exists
$\delta> 0$ such that for every $v\in \mathcal A$ with $\|w-v\|_{H^1}\le
\delta$ we have $J(w)\le J(v)$. In the important case $d>1$ we can prove
(see $\textbf{Appendix C}$)  that $w \in \mathcal A$ is a local minimizer of
$J$ if and only if it coincides with a solution $w_\zeta$ of the following
system:
\begin{equation}
\label{link11}\left\{
\begin{array}{ll} &  w''= 0\ \ \hbox{in}\ \ (0,\zeta )\\
&\\
&w(0)=d;\ w(\zeta )=1\end{array}\right. \left\{
\begin{array}{ll} &  \nu ^2 w''= \displaystyle w\ \ \hbox{in}\ \ (\zeta ,1)\\
&\\
& w(\zeta )=1,\ w'(1)= 0\end{array}\right. ,
\end{equation} at $\zeta=\bar\zeta$, where $\bar \zeta$ is a
local minimizer of the function $\hat J(\zeta)=J (w_\zeta)$. In the case $d<1$
there is only one minimum given by the solution of the Euler-Lagrange equation
$\nu ^2 w''=w$ in $(0,1)$ with the boundary conditions $w(0)=d$ and $w'(1)=0$.
In the special case $d=1$ there are two solutions, namely, the solution of the
differential equation $\nu ^2 w''=w$ and the homogeneous solution $w=1$ which
means that the detached set can be either empty or coincide with the whole
interval $(0,1)$.

The solution of the linear equations (\ref{link11}) can be computed explicitly. It is
easy to show that they are given exactly by the formulae (\ref{ucont}). This means
that the $n\rightarrow \infty$ limit of the metastable branches of the discrete model
coincides with the metastable solutions of the continuum model. Therefore besides
ensuring convergence of the global minima our pointwise limit also preserves the
local minimizers.

\section{Dynamic strategies}\label{mdcr1}\setcounter{equation}{0}

In the previous sections  we found for each value of the loading parameter $d$  a variety of the accessible metastable configurations. Suppose now that the value of
the loading parameter is changing quasi-statically. Then the choice of a particular
local minimum occupied by the system at each value of the loading parameter is
controlled by dynamics. Of a particular interest are the two evolutionary
strategies. The first one represents the vanishing viscosity limit of the
corresponding viscoelastic (overdamped) problem (e.g. \cite{PT7}). In this case the
system stays in a given metastable state till it becomes unstable ({\it maximum delay
convention}). The second strategy imposes that the system is always in the global
minimum of the energy ({\it Maxwell convention}). This behavior can be viewed as the
zero temperature limit of the hamiltonian (underdamped) dynamics.

\subsection{Viscosity solution}\label{mdcr}

Suppose that the parameter $d$ is monotonically increasing starting from the
value $d=0$ with no debonded springs (point $O$ in Fig.\ref{f3}). The `virgin'
branch becomes unstable when the {\it first} spring  detaches at $w_1=1$.
According to (\ref{br}) the decohesion starts at $d=d_2(0)=1$. The system then
switches to a new metastable branch and we assume that in this new branch the
only first spring, verifying $w_1=1$, detaches whereas all other springs still
remain in the elastic regime ($A-B$ in Fig.\ref{f3}).
 \begin {figure}[hbtp]
\vspace{0cm} \hspace{1.5cm}
\includegraphics[width=10 cm]{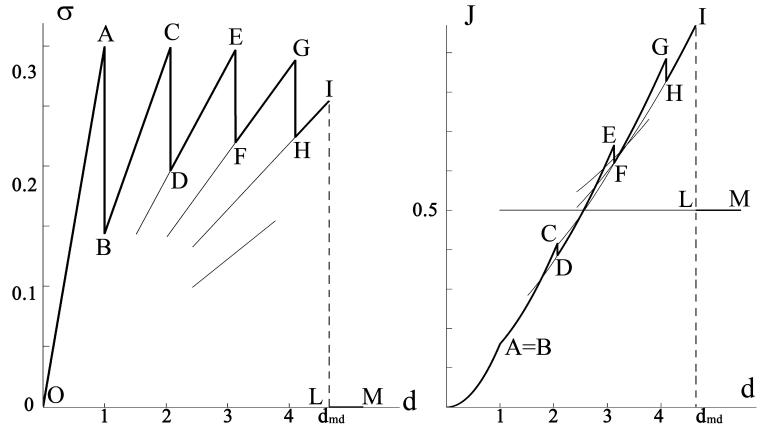}
\vspace{0 cm} \caption{\setlength{\baselineskip}{13 pt}{\label{f3} \footnotesize Overall force-displacement relation for a system with $n=6$ springs and $\nu=0.3$ Viscosity solution (maximum delay strategy)
is indicated by the bold lines.}}
\end{figure} To check that only one spring breaks one has to study
the global energy landscape and determine the steepest descending paths (see
e.g. \cite{PT5}). We observe that in the continuum limit, according with
(\ref{dddd}), a single metastable branch can be associated to each $d$ and this indeterminacy is
automatically overcome. If the
displacement is increased further, the debonding continues as the second spring
reaches the breaking limit at $d=d_2(1)$ ($C$ in Fig.\ref{f3}). This pattern
repeats itself as the subsequent springs debond one at a time. As we see the
system follows a `pinning-depinning' type of evolution with alternating slow
elastic stages and sudden transitions between different metastable branches.
This behavior, with the system switching between the branch with $\xi$ debonded
springs to the branch with $\xi+1$ debonded springs is possible till
$d_{2}(\xi+1)>d_{2}(\xi)$. The numerical solution shown in Fig.\ref{f3} shows
that there exists a value of the external load, $d=d_{md}$, such that for
$d>d_{md}$ the only equilibrium solution is the totally debonded configuration,
i.e. $\xi=n$. Thus, when $d=d_{md}$  all the remaining elastic springs break
simultaneously and the system jumps to the fully debonded configuration. In the
case of infinite $n$, we shall be able to find the value of the limiting
threshold $d_{md}$ analytically (see Section~\ref{cont}).

\subsection{Global minimization}

 To find the global minimum we have
to minimize the energy of the metastable equilibrium states with respect to the
parameter $\xi$. Fig.\ref{f5} shows by bold lines the Maxwell path for the same
system as in the previous section. We
observe the existence of another threshold $d_{Max}<d_{md}$ separating the regime
with a progressive debonding  from the sudden jump to the fully detached
configuration. Overall, the resulting stress-strain path is analogous to the one in
the case of the maximum delay convention. In quantitative terms, the Maxwell loading
path is lower and the transition to the fully debonded configuration is attained at a
lower assigned displacement.
\begin {figure}[hbtp]
\vspace{0cm} \hspace{1.5cm}
\includegraphics[width=11 cm]{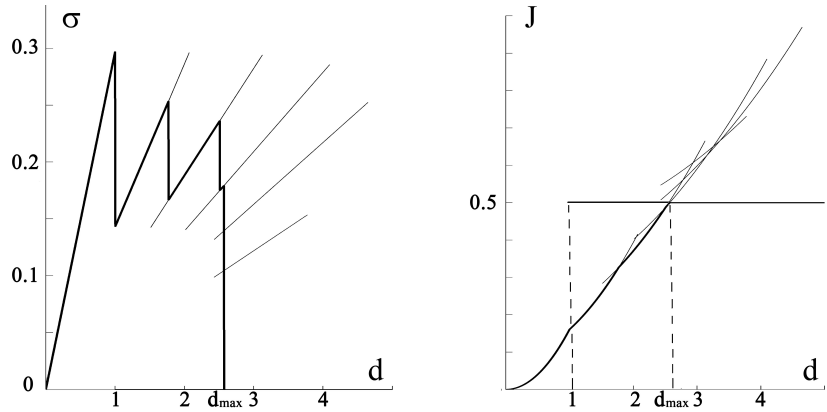}  \vspace{0 cm} \caption{\setlength{\baselineskip}{13
pt}{\footnotesize Overall force-displacement relation for a system with $n=6$ springs and $\nu=0.3$ Global minimum
solution (Maxwell convention) is indicated by the bold lines. Compare with Fig.\ref{f3}.\label{f5}}} \end{figure}

\subsection{Dissipation}
To characterize the dissipation associated with the maximum delay strategy, consider, for
instance, the path b-c-d-e shown in Fig.\ref{f5b}. We denote by $\bar \xi$ the number
of  debonded springs at the starting equilibrium branch b-c. The subsequent branch
d-e will then have $\bar \xi+1$ debonded elements. According to the maximum delay
convention, the system follows the equilibrium branch $\bar \xi$ until it becomes
unstable at $d=d_2(\bar \xi)$ (path b-c). Then it switches to a new branch with a
smaller energy (jump c-d). To find the energy dissipated during this jump event we
need to compare the energies (\ref{eqene}) corresponding to the two branches $\bar
\xi$ and $\bar \xi +1$ at a fixed displacement $d=\bar d$. We can write
 \be\Delta J(\bar d,\bar \xi) :=J_{\bar \xi+1}(\bar d)- J_{\bar \xi}(\bar d)
 =\frac{1}{2n}(n\bar d^{\,2}( K_{\bar \xi+1}-
 K_{\bar \xi})+1).\label{jump}\ee
\noindent The first term in the right hand side represents the difference of the
elastic energies (area inside the triangle O-P-Q in Fig.\ref{f5b}). The second
term represents the cohesive energy accumulated by the system in the transition
between the two states (it does not depend on $\bar \xi$). The released elastic
energy is partially accumulated by the system in the form of cohesion energy and
the rest is dissipated. The dissipation is zero when the energy difference in
(\ref{jump}) vanishes which corresponds to the Maxwell path. {\footnote{Under Maxwell convention the transition to the fully debonded state also  takes place
without dissipation and  at
$d=d_{Max}(\bar \xi_{Max})$  the decohesion energy equals the elastic
energy and  $\bar \xi_{Max} =n  K_{\bar
\xi_{Max}}d^2_{Max}.$} Instead, along the maximum delay path, represented by the points
b-c-d in Fig.\ref{f5b}, the system switches to the new branch in a dissipative way
(jump c-d) and the dissipated energy
 $\Delta J (d_2(\bar \xi),\bar \xi)
$ is equal to the area C-c-D-d. In general, according to the maximum delay convention
the area underneath the stress-strain path, representing the external work, can be
decomposed into the decohesion energy represented in Fig.\ref{f5b} by the equal
triangles of unit area (along the Maxwell
path at each switching event the increment of the decohesion energy has the same magnitude as the increment of elastic energy), the accumulated elastic energy
represented by the dark grey and the dissipated energy represented by the light areas
between the maximum delay path and the Maxwell path.

\begin {figure}[hbtp] \vspace{0cm} \hspace{1.5cm}
\includegraphics[width=8.5 cm]{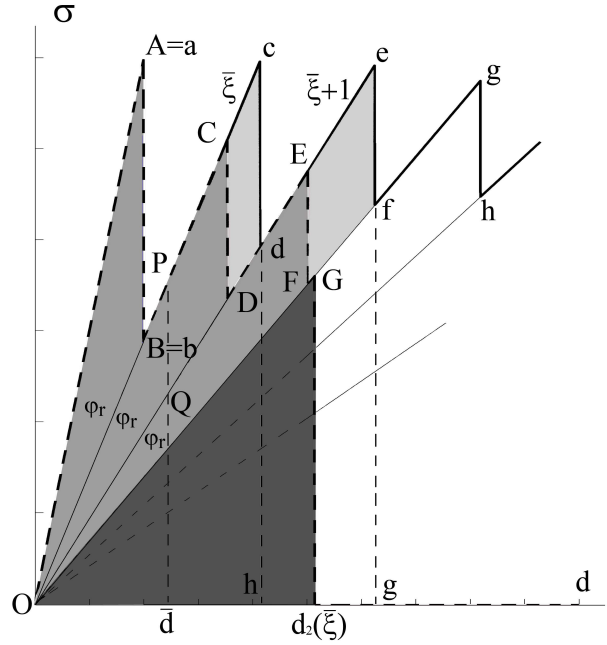} \vspace{-0.2 cm}
\caption{\label{f5b}\footnotesize Decomposition of the external work into the dissipated
energy (light grey areas C-c-D-d, E-e-F-f), the decohesion energy (middle grey area
O-A-B-C-D-E-F-O) and the elastic energy (dark gray area o-f-g-o); the global minimum response (dashed bold lines); the maximum delay response( continuous bold lines). Parameters are the same as in
Fig.\ref{f3}}
\end{figure}

\subsection{Hysteresis}

In Fig.\ref{f5f} and Fig.\ref{f5g} we illustrate the behavior of the system
under cyclic loading.  If the Maxwell convention is operative, there is
no hysteresis and the system follows elastically the same path for loading and
unloading (say, path O-A-B-A-O in Fig.\ref{f5f}$_a$).
\begin {figure}[hbtp] \vspace{0cm} \hspace{1.5cm}
\includegraphics[width=14.5 cm]{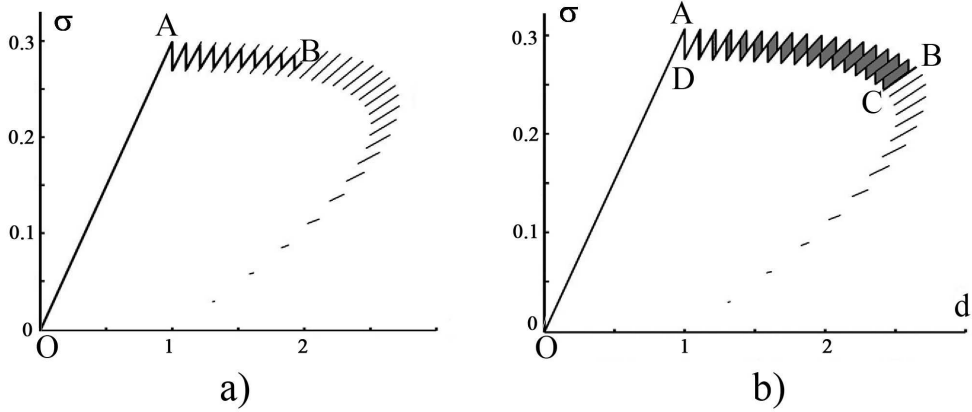} \vspace{0 cm}
\caption{\setlength{\baselineskip}{13 pt}{\footnotesize Partial cyclic loading for a
system with $n=30$, and $\nu=0.3$). a)
Maxwell convention, b) maximum delay convention. \label{f5f}}}
\end{figure}

If the system is unloaded, after the transition to the fully
debonded state ($d>d_{Max}$) the crack heals again at $d=d_{Max}$, with a sudden
transformation of the cohesive energy into elastic energy. During such event a
finite domain of broken springs reconnects simultaneously (path D-C-B in
Fig.\ref{f5g}$_a$).
Such snaps are indeed observed in experiments, both for loading
and unloading (see \cite[Chapter 3]{Ken} and references therein). Experiments show,
however, that the detachment and reattachment thresholds can be different with the
corresponding systems exhibiting an adhesion hysteresis. This suggests that the
Maxwell strategy may be less realistic than the maximum delay strategy.

\begin {figure}[hbtp] \vspace{0.5cm} \hspace{1.5cm}
\includegraphics[width=13.5 cm]{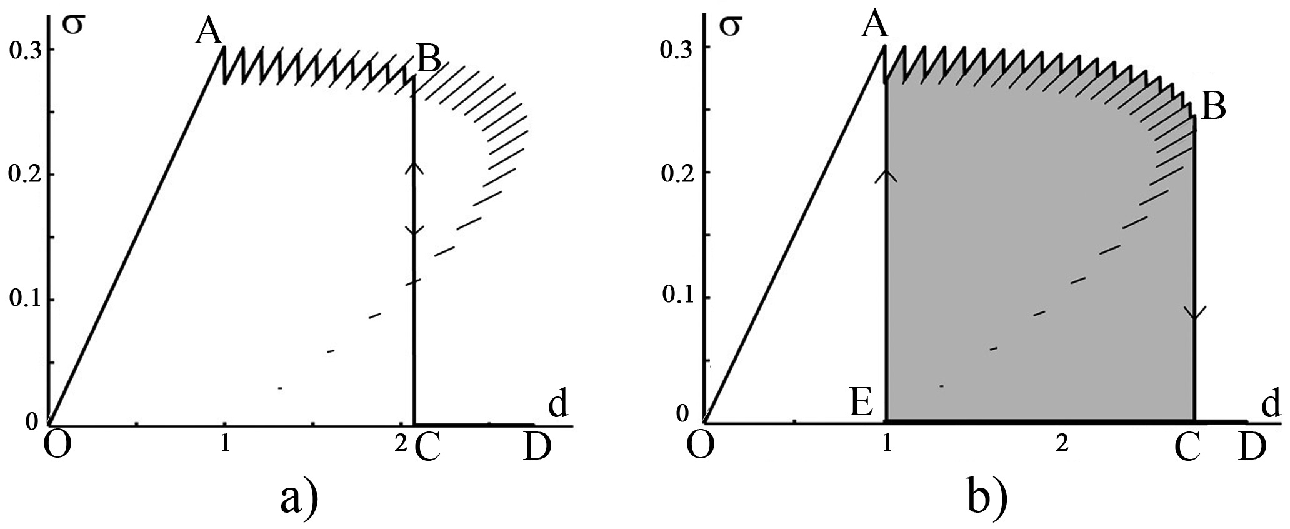} \vspace{0 cm}
\caption{\setlength{\baselineskip}{13 pt}{\footnotesize Complete cyclic loading for a
system with $n=30$, $\nu=0.3$.
a) Maxwell convention, b) maximum delay convention. \label{f5g}}}
\end{figure}

Under the maximum delay convention, if the unloading starts before the system
reached the fully debonded state ($d<d_{md}$), the system shows a limited
hysteresis (loop O-A-B-C-D-A in Fig.\ref{f5f}$_b$) which disappears in the
continuum limit.  If we unload the chain inside this hysteresis loop (from,
say, a branch $\bar\xi$), the system first deforms elastically until
$d=d_1(\bar\xi)$ when the last broken spring reconnects and the system jumps
back to the branch $\bar \xi-1$. With further unloading the system follows this
new equilibrium branch until again at $d=d_1(\bar \xi-1)$ another springs
reconnects and so on.

In Fig.\ref{f5g}$_b$ we illustrate the behavior of the dissipative system during the
complete unloading from the fully debonded state. We remark that in contrast to the
case of small cycle unloading  the maximal hysteresis is preserved in the macroscopic
limit.

\section{Continuum behavior}\label{fin}\setcounter{equation}{0}

We now turn to the study of the continuum solutions (\ref{ucont}). Using
(\ref{limd}) one can see that there exists a critical value $d_{md}$ such that
for $1<d<d_{md}$ the function $\hat J(\zeta) $ from (\ref{J}) has two
non-degenerate critical points where $\hat J'(\zeta)=0$ (one stable and one
unstable) while for $d>d_{md}$ there are no such critical points. \begin
{figure}[hbtp] \vspace{0cm}
\includegraphics[width=8 cm]{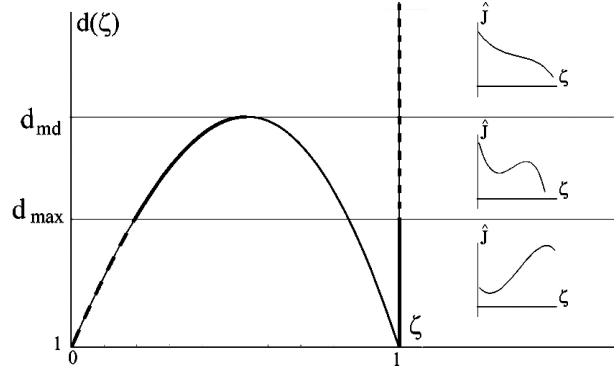} \vspace{0 cm}
\caption{\setlength{\baselineskip}{13 pt}{\footnotesize Phase diagram for the
continuum model. Bold lines represent the maximum delay convention;
bold-dashed lines represent the global minimization strategy (Maxwell
convention). Inserts show the structure of the function $\hat J(\zeta)$ in the
corresponding intervals. Here $\nu=1$.\label{JJ}}}
\end{figure}
\begin {figure}[hbtp]
\vspace{0cm}  \includegraphics[width=12.5 cm]{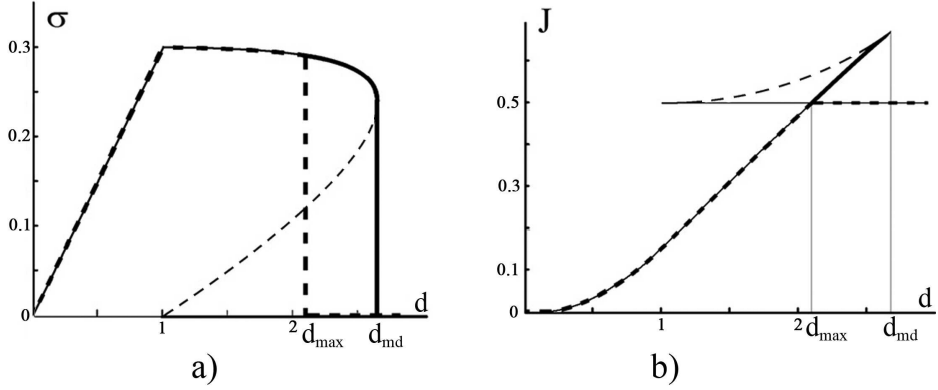} \vspace{0 cm}
\caption{\setlength{\baselineskip}{13 pt}{\footnotesize Equilibrium force and energy
for the continuum analog of the system considered in Fig.\ref{f5f}. Bold lines
represent the maximum delay convention; bold-dashed lines - the global minimization
strategy (Maxwell convention); bold continuous lines correspond to unstable equilibria.\label{f10}}}\end{figure}

One can also check that for $d>1$ the derivative $\hat J'(0)=\hat
J'(1)=1-d^2<0$, which means that the function $\hat J(\zeta)$ behaves as shown
in the inserts in Fig.\ref{JJ}. Notice also that there exists another threshold
$d_{max}$ such that for $d<d_{max}$ the global minimum is attained at the first
of the two critical points, whereas for $d_{max}<d<d_{md}$ the global minimum
is attained at the boundary of the domain, $\zeta=0$, describing the totally
debonded configuration. Moreover, this state remains the only minimizer for the
whole interval $d>d_{md}$. In Fig.\ref{f10} we show the stress-strain and
energy-strain diagrams illustrating this behavior of the continuum solutions.

The critical value of displacement $d_{md}$ can be obtained from the equation $
d'(\zeta)=0$ (see Fig.\ref{JJ}), which gives the fraction $\zeta_{md}$ of
debonded springs at $d=d_{md}$. We can  write explicitly  $$
\zeta_{md}\frac{\sigma^2(\zeta_{md})}{\nu^2}+ \sigma(\zeta_{md})-
\zeta_{md}=0.$$ After solving this equation, we can use (\ref{limd}) to find
$d_{md}= d(\zeta_{md}).$ The displacement $d_{max}$ can be obtained by first
determining the  fraction of debonded springs $\zeta_{max}$ which satisfy $\hat
J(\zeta_{max})=\hat J(1)$ or $$\zeta_{max} \frac{\sigma^2(\zeta_{max})}{\nu^2}+
\sigma(\zeta_{max})-(1-\zeta_{max})=0.$$ Then, using (\ref{limd}), one can find
$d_{max} = d(\zeta_{max}).$

The overall comparison of Fig. \ref{f5g} and Fig. \ref{f10} shows that the
discrete system has a much richer set of metastable states (local minima) than
its continuum analog. As $n$ increases, we observe two major tendencies: some
of the branches of the local minima of the discrete system shrink to points
representing the local minima of the continuum system whereas some other
branches simply disappear. On the contrary, the structure of the
global minimum path remains basically unaffected as $n\rightarrow\infty$.
\begin {figure}[hbtp]
\vspace{0cm}  \includegraphics[width=14 cm ]{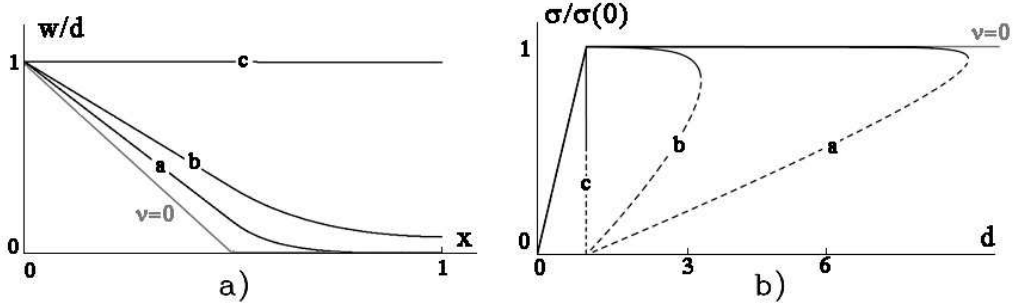} \vspace{0 cm}
\caption{\setlength{\baselineskip}{13 pt}{\footnotesize a) Displacement fields
at $\zeta=0.5$ and b) force-displacement diagrams: (a) $\nu=0.1$,  (b)
$\nu=0.25$, (c) $\nu=10$. Dashed lines in b) correspond to unstable equilibrium
configurations, the horizontal plateaux represents the thermodynamic limit
$\nu=0$.  \label{f14}}}
\end{figure}
\begin {figure}[hbtp]
\vspace{0cm}  \includegraphics[width=14 cm ]{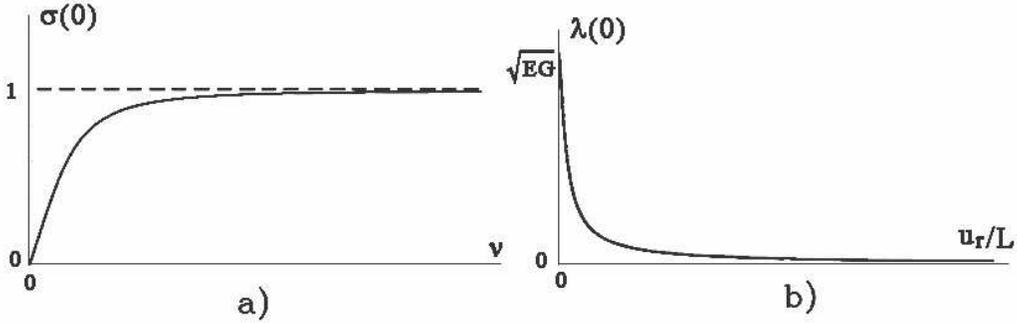} \vspace{0 cm}
\caption{\setlength{\baselineskip}{14 pt}{\footnotesize a) Normalization factor
 $\sigma(0)$ as a function of $\nu$.
b) Debonding force $\lambda(0)$ as a function of $u_r/L$. \label{f15}}}
\end{figure}

Due to the relative simplicity of the continuum model, one can study the
dependence of the response on the remaining nondimensional parameter
$\nu$ characterizing the toughness of the adhesion layer.  In Fig.\ref{f14}$_a$
we show three displacement fields corresponding to a given size of the crack
($\zeta=0.5$) and to different values of the nondimensional parameter $\nu$. We
observe that the `localization' of the crack tip predicted by our model
increases as $\nu$ decreases. In Fig.\ref{f14}$_b$ we represent the
 force-displacement diagrams generated by the continuum model at
different values of $\nu$. Here the force is normalized by
the debonding threshold, corresponding to $\zeta=0$
  $$\sigma(0)=\nu \tanh{\frac{1}{\nu}}.$$
As the parameter
$\nu$ decreases we observe an interesting evolution of the stress-strain
response from an instantaneous (brittle) debonding of the whole chain to a
(plasticity type) plateaux, in the case of a localized tip. In particular, one can see that the ductility of the system,
represented by the overall size of the adhesion hysteresis, grows as the
parameter $\nu$ decreases.

For the continuum system ($n=\infty$) the limit $\nu \rightarrow 0 $  can also be interpreted as a ``thermodynamic" limit because simultaneously $u_r/L \rightarrow 0$ and
$l/L \rightarrow 0$. In this case the normalization load vanishes, $$\nu \tanh{\frac{1}{\nu}}\rightarrow 0$$ however, the limit of the actual critical debonding force remains finite
$$\lambda(0)=(L/u_r)E\sigma(0)=\sqrt{EG}\tanh \left ( \sqrt{\frac
{E}{G}\frac{L}{u_r}} \right ). $$ In particular, the value of the critical force can be computed explicitly  $$\lambda_c=\sqrt{EG}.$$ In the context of the DNA denaturation, where our equilibrium metastable solutions represent domain walls connecting bonded and debonded states, the value $\lambda_c$ corresponds to the zero temperature unzipping threshold \cite{Theo}.

\section{Acknowledgements}
The work of G.P. was supported by the Progetto Strategico, Regione Puglia:
``Metodologie innovative per la modellazione e la sperimentazione sui materiali
e sulle strutture, finalizzate all'avanzamento dei sistemi produttivi nel
settore dell'Ingegneria Civile". The work of L.T.   was supported by the EU
contract MRTN-CT-2004-505226.

\section{Appendix A}

To invert the tri-diagonal matrix $\Bz_e$ in (\ref{B}) we can use iterative formulas
from \cite{Nabb}. We first relabel displacements as follows
 \be
\ba{lll} u_i:=w_i, & i=1,...,\xi, & \mbox{displacements of debonded springs} \\
v_j:=w_{j+\xi}, & j=1,...,n+1-\xi, & \mbox{displacements of elastic springs} \ea
\label{ef} \ee
 \noindent and define the vectors $\uuz = (u_i)$ and  $\vz =(v_j)$.
 Consider the first $\xi$ equations (\ref{eul2})
  corresponding to the debonded part
of the chain. After rearrangement, these equations can be rewritten as
 \be  \Ba \uuz\!=n\nu\!\left[ \!\!\ba{ccccccc}
    \!\! 2&\!\! -1 & \!\! &\!\! &\!\! & \!\!&\!\!\mbox{ \large {\bf 0}}\\
  \!\!  -1& \!\!2 & \!\!-1&\!\! & \!\!&\!\! &\!\! \\
  \!\!  &  \!\! & \!\!\ddots&\!\!\ddots&\!\!\ddots&\!\! &\!\! \\
    \!\!& \!\!& \!\!  &\!\! -1&\!\!2&\!\!-1&\!\!\\
   \!\! \mbox{ \large {\bf 0}}&\!\! &\!\!   &\!\!  &\!\! &\!\! -1 &\!\! 2\\
    \ea\! \right ]\!\!
    \left [ \ba{c}  u_1 \\
    u_2 \\ \ddots \\ u_{\xi-1} \\ u_{\xi} \ea \right ]\!\! =\!\!
    \left [  \ba{c}  \frac{\sigma}{\nu} + n\nu\,d \\
    0 \\ \ddots \\ 0 \\n\nu v_1 \ea \right ] \,,
 \label{fract}\ee
\noindent where we introduced the $\xi \times \xi$ matrix $\Ba$ and added $d$
to both sides of the first equation. The parameter $v_1$ is the
deformation of the first bonded spring. Observe that $\Ba$ is a Toeplitz tri-diagonal
matrix which can be inverted explicitly (see e.g. \cite{HC,Nabb})

 $$(\Ba^{-1})_{ij}=\frac{1}{\nu}\frac{(i+j-|j-i|)(2\xi+2-|j-i|-i-j)}
 {4n (\xi+1)}.$$
\noindent Since in the right hand side of (\ref{fract}) only the first and the last
elements are different from zero, we are interested only in
 $$(\Ba^{-1})_{i 1}=(\frac{1}{n \nu}-\frac{i}{n \nu (\xi+1)}), \hkkk
 (\Ba^{-1})_{i \xi}= \frac{i}{n \nu(\xi+1)}.$$
\noindent Using the first equation of (\ref{fract}) and (\ref{bc}) we obtain
 \be \sigma=n\nu^2\, \frac{d -v_1}{\xi}. \label{bcc}\ee
\noindent The remaining equations give
 \be u_i=d - (i-1)\frac{\sigma}{n \nu^2} , \hkkk i=1,...,\xi.
 \label{fp}\ee

Similarly,  we can reformulate the remaining $n+1-\xi$ equations corresponding to the
bonded part of the chain in the form

$$\Bb \vz\!=\left[\!\! \ba{ccccccc}
     2+\frac{1} {n^2\nu^2}&\!\! -1 &\!\!  &\!\! &\!\! &\!\! &\mbox{ \large {\bf 0}}\\
    -1&\!\! \!\!2+\frac{1} {n^2\nu^2}&\!\! -1&\!\! &\!\! &\!\! & \\
    &\!\!   &\!\! \ddots&\!\!\ddots&\!\!\ddots&\!\! & \\
    &\!\! &\!\!   &\!\! -1&\!\!2+\frac{1} {n^2\nu^2}&\!\! -1& \\
    \mbox{ \large {\bf 0}}& \!\!&  \!\! & \!\! & \!\!&\!\! -1 &\!\! 2+\frac{1} {n^2\nu^2}\\
    \ea\!\! \right ]
 \!\!  \left [ \ba{c} \!\! \!v_1\! \\
   \! v_2 \! \\ \!\!\ddots \\\!\!   \\\!\! \!v_{n+1-\xi}\! \ea\!\! \right ]\! =
  \!\!  \left [  \ba{c} \!\!\!  d-(\xi-1)\frac{1}{n\nu^2}\sigma\!\! \\
    0 \\ \ddots \\ 0 \\ \!\! \!\!  v_{n+1-\xi} \ea \!\!\right ],
 $$
 \smallskip
\noindent where we introduced the $(n+1-\xi) \times (n+1-\xi)$ matrix $\Bb$. Once
again we have a tri-diagonal Toeplitz matrix and since the diagonal elements satisfy
$(\Bb)_{i i}>2$, $i=1,...,n+1   -\xi$, we can write (see again \cite{HC})
 $$(\Bb^{-1})_{i j}=\frac{\cosh[(n-\xi+2-|j-i|)\eta]-\cosh
 [(n-\xi+2-i-j)\eta]}{2  \sinh [\eta] \sinh [(n-\xi+2)
 \eta]}.$$
\noindent The parameter $\eta$ is given by the equation (\ref{eta}) and the results do
not depend on the choice of one of the two solutions of this equation (indeed the equilibrium solutions are even functions of $\eta$).
As in the previous case, we need only the first and the last columns of the inverse
matrix
 $$(\Bb^{-1})_{ i 1}=\frac{\sinh[(n-\xi+2-i)\eta]}
 { \sinh [(n-\xi+2)\eta]},$$
$$(\Bb^{-1})_{ i  (n-\xi+1)}=\frac{\sinh [i\, \eta]}
 {\sinh [(n-\xi+2)\eta]}.$$
\noindent By using these formulas we obtain
 \be\ba{lll} v_i&=&\displaystyle \frac{\sinh[(n+2-\xi-i)\eta]}
 {  \sinh [(n+2-\xi)\eta]}( d -
  \frac{1}{\nu^2}\frac{(\xi-1)\sigma}{n})+ \vks\\&+&\displaystyle
 \frac{\sinh [i\eta]}
 { \sinh [(n+2-\xi)\eta]}v_{n+1-\xi}, \hkkk \hkkk  i=1,...,n+1-\xi.\ea
 \label{tre}\ee

\section{Appendix B}\setcounter{equation}{0}

Here we prove that the energy $J(w)$ given by (\ref{link1}) is actually the
$\Gamma$-limit (with respect to the convergence in $L^2(0,1)$) of $J_n$ given
by (\ref{J11}) as $n\rightarrow \infty$.

To prove the $\Gamma$-convergence we proceed in several steps. The first step
is to find a tight lower bound. The following reasoning is standard (see
\cite{DM}).

\begin{prop}
\label{proposemi} Assume that $w_n\in {\mathcal A}^*_n$ and $J_n(w_n)\le C$ for
every $n\in \mathbb N$. Then up to subsequences $w_n\to w$ weakly in $H^1(0,1)$
and $w\in \mathcal A$.
 Moreover if $w_n\to w$ weakly in $H^1(0,1)$ and $w\in \mathcal A$ then

 \begin{equation}
 \label{semi}
 \displaystyle\liminf_{n\to\infty} J_n(w_n)\ge J(w).
 \end{equation}
 \end{prop}

 \begin{proof} The first assertion is trivial since
 $J_n(w_n)\le C$ implies
 $\|w'_n\|_{L^2}\le C'$
  which together with $w_n(0)=d$ yields weak
 compactness of $w_n$ in $H^1(0,1)$ and that any limit point of $w_n$
 belongs to $\mathcal A$. If in addition $w_n\to w$ weakly in
 $H^1(0,1)$ and $w\in \mathcal A$, then the convexity of $\phi$
 yields
  \[\displaystyle\liminf_{n\to\infty}\int_0^1\phi(\frac{u_r}{L}w'_n)\,dx\ge
  \int_0^1\phi(\frac{u_r}{L}w')\,dx.\]
  Moreover by recalling that $w_n\to w$ in each $L^p(0,1)$ and that
  \[\displaystyle\sum_{i=1}^{n+1}{\frac {1}{ n}}
  \varphi\left(w_n\left ({\frac{i } {n}}\right )\right)\]
  are the Riemann sums of the function
  $\varphi(w)$ we get\[\displaystyle\sum_{i=1}^{n}{\frac{1}{ n}}
  \varphi\left(w_n\left ({\frac {i} {n}}\right )\right)\to\int_0^1\varphi(w)\,dx\]
  thus proving the inequality (\ref{semi}).\end{proof}
  \noindent The next step is to prove the existence of a recovery sequence.
  \begin{prop}
  \label{recoprop}
  Assume that $w\in\mathcal A$.
  Then there exists a sequence $w_n\in\mathcal A^*_n$ such that $w_n\to w$ weakly in $H^1$ and
  \begin{equation}
  \label{recovery}J_n(w_n)\to J(w).
  \end{equation}
  \end{prop}
  \begin{proof} Let $w\in\mathcal A$ and define $w_n\in\mathcal A^*_n$ as in (\ref{link}).
  It is readily seen by convexity that $J_n(w_n)\to J(w)$.
  \end{proof}
\noindent Previous statements prove that $\displaystyle
J=\Gamma\!\!-\!\!\!\lim_{n\rightarrow \infty} J_n$. The relationship between
the minimization problems concerning the functionals $ J_n$ and $ J$ is
clarified in the next theorem.

\begin{theorem}
Let $\bar w_n\in\mathcal A^*_n$ such that
\begin{equation}
\label{minimizing} J_n(\bar w_n)-\inf_{\mathcal A^*_n} J_n\to 0,
\end{equation}

\noindent then up to subsequences $\bar w_n \to \bar w$ weakly in $H^1(0,1)$
and
$$J_n(\bar w_n)\to J(\bar w)=\inf_{\mathcal A} J.$$
\end{theorem}

\begin{proof}
It is readily seen that (\ref{minimizing}) yields  $J_n(\bar w_n)\le C$ for
suitable $C>0$ and by Proposition \ref{proposemi} we get, up to subsequences,
$\bar w_n \to \bar w$ weakly in $H^1(0,1)$ and
\begin{equation}
\label{semi2} \displaystyle\liminf_{n\to\infty} J_n(\bar w_n)\ge J(\bar w).
\end{equation}
Let now $w\in\mathcal A$. Then by Proposition \ref{recoprop} there exists a
sequence $w_n\in \mathcal A^*_n$ such that $J_n(w_n)\to J(w)$ and $w_n\to w$
weakly in $H^1$. Then either $J_n(w_n)\ge J_n(\bar w_n)$  or\[\displaystyle
J_n( w_n)-\inf_{\mathcal A^*_n} J_n\to 0\] and so
\[\displaystyle \,J(w)\ge \liminf_{n\to\infty} J_n( w_n)\ge\liminf_{n\to\infty} J_n(\bar w_n)\ge J(\bar w)\]that is $J(\bar w)=\min J$ and $\inf_{\mathcal A^*_n}J_n\to \min J$.

\end{proof}

\section{Appendix C}\setcounter{equation}{0}

\noindent{ We recall that $w\in \mathcal A$ is a local minimizer of $J$ if
there exists $\delta> 0$ such that for every $v\in \mathcal A$ with
$\|w-v\|_{H^1}\le \delta$ we have $J(w)\le J(v)$.

 We first show the following result.

 \begin{theorem} If  $w$
is a local minimizer of $J$ then:\begin{enumerate}
\item If $d< 1$ then
$\{w> 1\}=\emptyset$;\\
\item If $d> 1$ then either
$\{w> 1\}=[0,1]$ or  $\{w> 1\}=(0,\zeta ]$ with $\zeta\in (0,1)$.\\
\end{enumerate}\label{TTT}
\end{theorem}

\begin{proof}
\noindent  We begin with (2). Let $d> 1$, then $\{w>1\}$ is a non empty
relatively open subset of $[0,1]$ and therefore there exists a countable
collection of disjoint open intervals of $\mathbb R$, say $I_j,\ j\in \mathbb
N$ such that
$$
\{w>1\}= \bigcup_{j} (I_j\cap [0,1]).
$$
Assume by contradiction that for every $\zeta \in (0,1), \ \ \{w> 1\}\not = [0,
\zeta )$, then one of the following conditions holds true \par\medskip i) $\{w
> 1\}= [0,1]$\par\medskip ii) $\exists\ \alpha\in (0,1)$
 such that $(\alpha, 1]\subset
\{w>1\}$ and $w(\alpha)=1$\par\medskip iii) $\exists\ \beta,\ \gamma\in (0,1)$
such that $(\beta,\gamma)\subset \{w>1\}$ and $w(\beta)= w(\gamma)=1.$
\par\medskip
\noindent If  ii) holds then let $\eta\in C^1_0(0,1),\ \ \eta\equiv 0$ in
$[0,\alpha]$: since $w$ is a local minimizer we get for every $\varepsilon> 0$
such that $\varepsilon \|\eta\|_{H^1}< \delta$
$$
0\le J(w+\varepsilon\eta)-J(w)= \varepsilon\int_\alpha^1 (\nu^2 w' \eta'+
w\eta{\bf 1}_{\{w+\varepsilon\eta\le 1\}}) \,dx+ o(\varepsilon)
$$
and by
letting $\varepsilon\to 0$ we have
$$
\int_\alpha^1  w' \eta'=0
$$
that is $ w'' = 0$ in $(\alpha,1)$. Now, since $w(\alpha)=1$ and due to the
natural boundary condition $ w'(1)=0,\ $ we get  $w\equiv 1$ in the whole
$(\alpha,1)$, which is a contradiction. Case iii) follows analogously and hence
2) is proven. In order to prove 1) suppose by contradiction that $\{w > 1\}\not
= \emptyset$ with $d< 1$. Then either ii) or iii) holds true and a
contradiction can be obtained also in this case.\end{proof}

We can now study the relation between the local minimizers of the continuum problem
(\ref{link1}) and the solution of the linear system (\ref{link11}).

\begin{theorem}  $w\in \mathcal A$ is a local minimizer of $J$ defined by
(\ref{link1}) if and only if
there exists $\bar\zeta \in (0,1)$ such that \be \left\{
\begin{array}{ll} &  w''= 0  \hbox{ in}\ \ (0,\bar\zeta )\\
&\\
& w(0)=d;\ w(\bar\zeta )=1\end{array}\right. \label{Fr1}\ee and \be\left\{
\begin{array}{ll} &  \nu ^2 w''=\displaystyle w\
 \ \hbox{in}\ \ (\bar\zeta ,1)\\
&\\
& w(\bar\zeta )=1,\ w'(1)= 0\end{array}\right.\label{Fr2}\ee and $\bar \zeta$
is a local minimizer of $\hat J(\zeta)=J(w_\zeta)$.\end{theorem}

\begin{proof} By Theorem \ref{TTT} we have that if $w\in \mathcal A$ is a local minimizer of
$J$ then there exists $\bar \zeta$ such that $w$ satisfies (\ref{Fr1}) and
(\ref{Fr2}). Moreover, given $\delta>0$, there exists a given small enough
$\varepsilon>0$, such that $w_{\varepsilon}$, the unique solution of
$$\left\{
\begin{array}{ll} &  w''= 0\ \ \hbox{in}\ \
(0,\bar\zeta
-\varepsilon)\\ &\\
& w(0)=d;\ w(\bar\zeta -\varepsilon)=1\end{array}\right.$$ and
$$\left\{
\begin{array}{ll} & \nu ^2 w''= \displaystyle w\
\ \hbox{in}\ \ (\bar\zeta -\varepsilon,1)\\
&\\
& w(\bar\zeta -\varepsilon)=1,\ w'(1)=0 ,\end{array}\right.$$ satisfies
$\|w_\varepsilon-w\|_{H^1}< \delta$. This follows from well known results for
elliptic equations with variable domains (see \cite{BB}). Hence
$J(w_\varepsilon)\le J(w)$ and therefore $\bar\zeta $ is a local minimizer of
the function $\hat J(\zeta)= J(w_\zeta)$.

To prove the inverse statement we have to show that if $\bar\zeta$
 is a local minimizer for $\hat J$ then $w_{\bar\zeta}$ is a local
minimizer for $J$. Let $\eta>0$ such that for every $|\zeta-\bar\zeta|< \eta, \
\hat J(\bar\zeta)\le \hat J(\zeta)$: we may choose $\beta> 0$ such that if
$v\in H^1(0,1), \ v(0)=0,\ \|v\|_{H^1}\le\beta$, then $w_{\bar\zeta}+v> 1$ in
$[0,\bar\zeta -{\frac \eta 2})$ and $w_{\bar\zeta}+v< 1$ in $(\bar\zeta +{\frac
\eta 2},1]$. Hence
$$\begin{array}{lll}
&\displaystyle J(w_{\bar\zeta}+v)&\displaystyle \ge{\frac 1 2}
\int_0^{\bar\zeta -{\frac \eta 2}}\left (\nu^2| w'_{\bar\zeta}+ v'|^2+ 1\right
)\,dx+{\frac 1 2}\int_{\bar\zeta -{\frac \eta 2}}^{\bar\zeta +{\frac \eta
2}}(\nu^2| \hat u'|^2+ |\hat u\wedge 1|^2)\,dx+\vspace{0.2 cm}\\
&\displaystyle &+\displaystyle{\frac 1 2}\int_{\bar\zeta +{\frac \eta
2}}^1(\nu^2| w'_{\bar\zeta}+ v'|^2+ |w_{\bar\zeta}+v|^2)\,dx
\end{array}
$$
where $\hat u(x)\wedge 1=\hat u(x )$ if $u(x)<1$ and $\hat u(x)\wedge 1=1$ if
$\hat u(x)\geq 1$. Here $\hat u$ denotes an absolute minimizer of
$$
u\to {\frac 1 2}\int_{\bar\zeta -{\frac \eta 2}}^{\bar\zeta +{\frac \eta
2}}(\nu^2|u'|^2+|u\wedge 1|^2)\,dx
$$
among all
$u\in H^1(\bar\zeta -{\frac \eta 2},\bar\zeta +{\frac \eta 2})$ such that
$u({\bar\zeta \pm{\frac \eta 2}})=w_{\bar\zeta}(\bar\zeta \pm{\frac \eta {2
}})+v(\bar\zeta \pm{\frac \eta {2 }})$.

Therefore by defining
$$ w^*(x)=\left\{
\begin{array}{ll} & w_{\bar\zeta}(x)+v(x) \ \ \hbox{in}\ \ [0,1]\setminus [\bar\zeta -{\frac \eta
2},\bar\zeta +{\frac \eta
2}]\\
&\\
& u(x)\ \ \hbox{otherwise}\end{array}\right.$$ we get
$$
J(w_{\bar\zeta}+v)\ge J(w^*)\ge J(w_{\zeta^*}).
$$

 An argument very close to that used in
 the beginning of this Appendix shows that there exists a unique $\zeta^*\in
(\bar\zeta-{\frac \eta {2 }},\bar\zeta+{\frac \eta {2 }})$ such that $u> 1$ in
$(\bar\zeta -{\frac \eta 2},\zeta^*)$ and  $u< 1$ in $(\zeta^*, \bar\zeta
+{\frac \eta 2})$. Then, taking into account that
$$
{w_{\zeta^*}}_{|_{(0,\zeta^*)}}\!\!\in\! \argmin \!
\left\{\!\!\int_0^{\zeta^*}\!\!\!\!\!\!(\nu^2| w'|^2+|w\!\wedge \!1|^2)dx: w\in
H^1(0,\zeta^* ), w(\zeta^* )\!=\!1, \!w(0)\!=\!d\!\right\}
$$
and
 $$
{w_{\zeta^*}}_{|_{(\zeta^*,1)}}\!\!\in\argmin\!\left\{\!\int_{\zeta^*}^1(\nu^2|w'|^2+|w\wedge
1|^2)dx: w\in H^1(\zeta^*,1), w(\zeta^*)\!=\!1\right\},
$$
 since
$\bar\zeta$ is a local minimizer for $J$ and $|\bar\zeta-\zeta^*|\le\eta/2$ ,
we argue
$$
 J(w_{\zeta^*})=J(\zeta^*)\ge J(\bar\zeta)=J(w_{\bar\zeta}),
$$
thus proving the local minimality of $w_{\bar\zeta}$.

\end{proof}

\end{document}